\documentclass[11pt]{article}

\usepackage[letterpaper,margin=1.15in]{geometry}
\usepackage[T1]{fontenc}
\IfFileExists{newtxtext.sty}{\usepackage{newtxtext}}{}
\usepackage{amsmath,amsthm,mathtools}
\numberwithin{equation}{section}
\IfFileExists{newtxmath.sty}{\usepackage{newtxmath}}{\usepackage{mathptmx}\usepackage{amssymb}}
\usepackage{microtype}
\microtypesetup{nopatch=footnote}
\usepackage{booktabs}
\usepackage{tabularx}
\usepackage{array}
\usepackage{enumitem}
\usepackage{tikz}
\usetikzlibrary{arrows.meta,calc,positioning,shapes.geometric}
\usepackage{pgfplots}
\pgfplotsset{compat=1.18}
\usepackage[hidelinks,hyperfootnotes=false]{hyperref}
\usepackage{setspace}
\usepackage{aliascnt}
\newtheorem{theorem}{Theorem}[section]

\newaliascnt{lemma}{theorem}
\newtheorem{lemma}[lemma]{Lemma}
\aliascntresetthe{lemma}

\newaliascnt{corollary}{theorem}
\newtheorem{corollary}[corollary]{Corollary}
\aliascntresetthe{corollary}

\newaliascnt{proposition}{theorem}

\aliascntresetthe{proposition}

\newaliascnt{definition}{theorem}

\aliascntresetthe{definition}

\newaliascnt{remark}{theorem}
\newtheorem{remark}[remark]{Remark}
\aliascntresetthe{remark}

\usepackage[nameinlink,capitalise,noabbrev]{cleveref}
\crefname{theorem}{Theorem}{Theorems}
\crefname{lemma}{Lemma}{Lemmas}
\crefname{corollary}{Corollary}{Corollaries}
\crefname{proposition}{Proposition}{Propositions}
\crefname{definition}{Definition}{Definitions}
\crefname{remark}{Remark}{Remarks}

\usepackage{xcolor}
\definecolor{RootA}{RGB}{0,72,140}
\definecolor{RootB}{RGB}{170,38,38}
\definecolor{RootC}{RGB}{0,110,70}
\definecolor{Slate}{RGB}{95,102,112}
\definecolor{TintA}{RGB}{232,240,250}
\definecolor{TintC}{RGB}{231,245,238}

\newcommand{\dens}{\operatorname{dens}}
\newcommand{\Round}{\textup{\textsc{Round}}}
\newcommand{\OPT}{\operatorname{OPT}}
\newcommand{\calP}{\mathcal P}
\newcommand{\wideG}{\widehat G}
\newcommand{\wideE}{\widehat E}
\newcommand{\mst}{\operatorname{mst}}
\title{\bfseries An $8/5$ Rounding for Half-Integral Forest-BCR\\[.15em]
via Root Supports and Circuit Rank}
\author{\textbf{Morteza Alimi}\\[-.05em]
\textit{University of Augsburg, Augsburg, Germany}\\[-.05em]
\texttt{morteza.alimi@uni-a.de}}
\date{}

\hypersetup{
  pdftitle={An 8/5 Rounding for Half-Integral Forest-BCR via Root Supports and Circuit Rank},
  pdfauthor={Morteza Alimi},
  pdfsubject={Approximation algorithms; Steiner Forest; bidirected cut relaxation},
  pdfkeywords={Steiner Forest, bidirected cut relaxation, LP rounding, circuit rank}
}

\begin{document}
\maketitle

\begin{abstract}
We study the rounding of a supplied half-integral feasible solution of the
root-assignment bidirected cut relaxation for Steiner Forest (Forest-BCR).
Byrka, Grandoni, and Traub [IPCO 2025] proved a $16/9$ guarantee for a recursive
framework that normalizes the LP point, selects a vertex set of maximum
projected LP density, buys a minimum spanning tree on that set, contracts it,
and recurses.  We prove that the same framework has guarantee
$8/5$.

The new analysis keeps the orientation and the root label of each projected
half-unit of LP mass.  In a simple projection, the cut constraints at a
terminal of degree two determine the root-assignment vector of every demand
incident with it, and half-integrality leaves only two possibilities: a unit
assignment to one root, which forces excess outdegree inside that root's
support, or a split assignment to two roots, which forces overlap between
their supports.  For every connected component $C$ of the split-root graph
this yields
$ \beta_C\ \ge\ \frac{L_C}{2}$,
where $\beta_C$ is the circuit rank of the union of the root supports in $C$
and $L_C$ is the number of its vertices of degree two in the full projection.
Balancing the density certificate obtained from this inequality against the
ordinary degree sum gives a vertex set of density at least $5/8$, and the
inherited contraction lemma turns that into the $8/5$ rounding.  

For every $q\ge3$ we also construct a normalized half-integral point whose
maximum projected density is exactly $5q/[2(4q-1)]$, so the universal
projected-density bound is asymptotically tight.  
\end{abstract}


\section{Introduction}\label{sec:introduction}

In the \emph{Steiner Forest} problem, an undirected graph with nonnegative
edge costs and a finite list of demand pairs are given, and the goal is to buy
a minimum-cost subgraph connecting the endpoints of every demand.  The
classical primal--dual algorithms of Agrawal, Klein, and Ravi and of Goemans
and Williamson give a factor-$2$ approximation~\cite{AKR95,GW95}; Jain's
iterative-rounding theorem gives another factor-$2$ route through a much more
general network-design relaxation~\cite{Jain01}.  After more than three
decades, Ahmadi, Gholami, Hajiaghayi, Jabbarzade, and Mahdavi obtained the
first factor below $2$~\cite{AGHJM25}.  Gupta and Traub subsequently gave a
simpler $1.994$-approximation~\cite{GT26}.

This paper concerns a different, LP-relative question.  Bidirected cut
formulations have long played a central role for Steiner Tree
\cite{GM93,Wong84,RV99}; for a broader comparison of Steiner Forest
formulations, see~\cite{SZM21}.  The bidirected cut relaxation for
Steiner Tree has a polynomial-size variable representation and
polynomial-time cut separation, and its integrality gap was recently proved to
be below $2$~\cite{BGT24}; a recent preprint gives the bound
$1.898$~\cite{PT26}.  Byrka,
Grandoni, and Traub introduced a root-assignment extension to Steiner Forest,
which we call Forest-BCR throughout~\cite{BGT25,BGT26}.  Every demand is
fractionally assigned to candidate roots,
and the mass assigned to a root is routed toward that root.  They proved an
unrestricted integrality-gap lower bound of $3/2$ and a $16/9$ rounding theorem
for every half-integral feasible point.

The half-integral result does not by itself give a $16/9$-approximation for
Steiner Forest, because an optimum Forest-BCR point need not be half-integral.
Nevertheless, half-integral points are a natural structural test for this
compact relaxation.  The known asymptotic $3/2$ lower-bound family does not
settle the restricted question, because on those instances the optimal
Forest-BCR solutions are not half-integral~\cite{BGT26}.  Its $q=2$
member does give a half-integral
bound: the graph has six vertices, so every feasible forest is a spanning
tree of cost $5$, while the displayed point of cost $2q=4$ has all
coordinates in $\{0,\tfrac12,1\}$ because $1/q=\tfrac12$.  The half-integral
rounding gap is therefore at least $5/4$~\cite[proof of Thm.~4]{BGT26}.

The recursive framework of~\cite{BGT26} is especially simple.  After metric
preprocessing and normalization, it projects each half-unit of LP mass to an
undirected multiedge, selects a vertex set of maximum LP density, buys a
minimum spanning tree on that set, contracts it, and recurses.  Their
$16/9$ bound follows from a universal density lower bound of $9/16$
(\cite[Lem.~5]{BGT26}).  We strengthen that one lemma to $5/8$ and change
nothing else.

\subsection{Main result and structural idea}

\begin{theorem}[Main theorem]\label{thm:main}
Let $G$ be a finite undirected graph with nonnegative rational edge costs
encoded in binary and let $\calP$ be a finite list of demand pairs.  There is
a deterministic algorithm, polynomial in the input encoding length, that,
given an explicitly represented half-integral feasible Forest-BCR point
$(x,z)$, outputs a feasible Steiner forest $F$ satisfying
\[
  c(F)\le \frac85\,c(x).
\]
\end{theorem}

It is convenient to name the quantity that \cref{thm:main} bounds.  The
\emph{half-integral rounding gap} is the supremum of $\OPT(I)/c(x)$ over all
instances $I$ and all half-integral feasible Forest-BCR points $(x,z)$ for $I$
with $c(x)>0$, where $\OPT(I)$ is the minimum cost of a feasible Steiner
forest.  \Cref{thm:main} shows that this quantity is at most $8/5$, and the
$q=2$ instance recalled above shows that it is at least $5/4$.

The contribution is the structural analysis of the normalized point.  At a
globally projected degree-two terminal $v$, the singleton and complementary
singleton cut constraints account for the entire unit of incident LP mass.
They therefore determine the root-assignment vector of every demand incident
with $v$.  Half-integrality leaves only two possibilities:
\begin{itemize}[leftmargin=1.8em,itemsep=.2em]
\item a unit assignment to one root, which contributes excess outdegree
      inside one root support; or
\item a split assignment to two roots, which contributes overlap between
      their two supports.
\end{itemize}
For a connected component $C$ of the split-root graph, let $L_C$ denote the
number of globally degree-two support vertices in the union of its root
supports and let $\beta_C$ denote the circuit rank of that union.  Combining
internal excess outdegree with overlap gives the key inequality
\[
  \beta_C\ge\frac{L_C}{2}.
\]

\subsection{Why the density is \texorpdfstring{$5/8$}{5/8}}
\label{sec:balance-intro}

The numerical optimization is short; Sections~\ref{sec:root-structure}
and~\ref{sec:cycle-rank} prove the structural statements that justify it.
Fix $\theta\in(0,1)$.  If some split-root component has
$L_C\ge\theta m_C$, where $m_C$ is the number of its labeled projected
edges, then connectedness and $\beta_C\ge L_C/2$ give a vertex set of density
at least
\begin{equation}\label{eq:intro-local}
  \frac{1}{2-\theta}.
\end{equation}
If no component satisfies this inequality, globally low vertices are scarce.
Since low vertices have projected degree two and all other support vertices
have degree at least three, the global degree sum gives density strictly
larger than
\begin{equation}\label{eq:intro-global}
  \frac{3}{2(2+\theta)}.
\end{equation}
The first expression increases in $\theta$ and the second decreases.  They
agree at
\begin{equation}\label{eq:crossing}
  \theta=\frac25,
  \qquad
  \frac{1}{2-\theta}=\frac{3}{2(2+\theta)}=\frac58.
\end{equation}
Thus every normalized point has a set of density at least $5/8$, and the
contraction lemma yields factor $8/5$; see \cref{fig:crossing}.

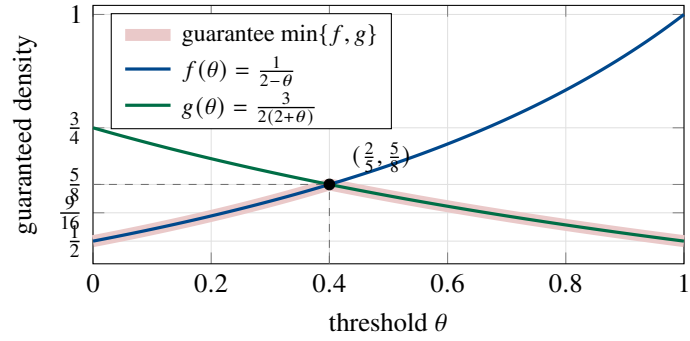
\begin{figure}[t]
\centering
\begin{tikzpicture}
\begin{axis}[
  width=9.4cm,height=5.0cm,
  xlabel={threshold $\theta$},
  ylabel={guaranteed density},
  xmin=0,xmax=1,ymin=0.45,ymax=1.02,
  xtick={0,0.2,0.4,0.6,0.8,1},
  ytick={0.5,0.5625,0.625,0.75,1},
  yticklabels={$\tfrac12$,$\tfrac9{16}$,$\tfrac58$,$\tfrac34$,$1$},
  grid=both,grid style={line width=.2pt,draw=black!12},
  legend pos=north west,legend cell align=left,
  label style={font=\small},tick label style={font=\small},
  legend style={font=\footnotesize}]
\addplot[RootB!25,line width=4.5pt,domain=0:0.4,samples=50,forget plot]{1/(2-x)};
\addplot[RootB!25,line width=4.5pt,domain=0.4:1,samples=50]{3/(2*(2+x))};
\addlegendentry{guarantee $\min\{f,g\}$}
\addplot[RootA,line width=1.2pt,domain=0:1,samples=80]{1/(2-x)};
\addlegendentry{$f(\theta)=\frac1{2-\theta}$}
\addplot[RootC,line width=1.2pt,domain=0:1,samples=80]{3/(2*(2+x))};
\addlegendentry{$g(\theta)=\frac3{2(2+\theta)}$}
\addplot[only marks,mark=*,mark size=2pt,black] coordinates {(0.4,0.625)};
\draw[dashed,black!55] (axis cs:0.4,0.45)--(axis cs:0.4,0.625);
\draw[dashed,black!55] (axis cs:0,0.625)--(axis cs:0.4,0.625);
\node[anchor=south west,font=\footnotesize] at (axis cs:0.42,0.63)
  {$(\tfrac25,\tfrac58)$};
\end{axis}
\end{tikzpicture}
\caption{The constant $5/8$ is the crossing point of the component
circuit-rank certificate and the global degree certificate.}
\label{fig:crossing}
\end{figure}

\paragraph{Contributions.}
The first contribution is the structural lemma $\beta_C\ge L_C/2$, which
converts every projected degree-two vertex into circuit rank.  The second is
that this lemma yields the universal $5/8$ projected-density bound, and hence
the $8/5$ rounding theorem, inside the inherited recursive
maximum-density/MST-contraction framework.  The third is a family with exact
maximum projected density $5q/[2(4q-1)]$: the universal projected-density
bound is therefore asymptotically tight, and on that family there exists a
legal execution of the framework with ratio $8/5-2/(5q)$ against the displayed
point.  The family is not an integrality-gap construction.  For completeness
and implementability, \cref{app:preprocessing} gives a self-contained
polynomial-time normalization.  Its exact feasibility test minimizes over
the two membership patterns in which a directed split can lower a valid-cut
capacity.  We claim no stronger normalization theorem than
\cite{BGT26}; the appendix only makes the implementation explicit.

\subsection{Relation to the previous proof}

The recursive rounding framework is inherited from Byrka, Grandoni, and
Traub~\cite{BGT26}: compute an exact maximum-density set, buy an MST,
contract, and recurse.  Our change is the structural analysis of normalized
half-integral points.  We retain the orientation and root label of each
projected half-unit, turning every globally degree-two vertex into an exact
circuit-rank or overlap contribution.  The previous proof
of~\cite[Lem.~5]{BGT26} uses an elegant unlabeled low/high-degree token
redistribution: every support vertex receives one token per incident
projected edge, a low vertex keeps its two, and a vertex of degree $k\ge3$
keeps $2+\frac{k-2}{k+1}\ge\frac94$ and passes $\frac{k-2}{k+1}\ge\frac14$ to
each neighbour, so that every support vertex ends with at least $9/4$ and the
density is at least $9/16$.  The root-labeled accounting gives $5/8$ instead.  Our \cref{thm:density} is a
drop-in replacement for that lemma: the hypotheses and the conclusion have the
same shape, and every other component of their analysis is used verbatim.

\Cref{app:preprocessing} reaches the same target normal form by a
self-contained implementation.  A direct calculation shows that a directed
split can reduce cut capacity in two membership patterns, and our feasibility
test enforces both.  We treat this appendix as infrastructure rather than as a
separate contribution.

\subsection{What is simple here, and what is not}
\label{sec:simplicity}

The rounding algorithm is unchanged, so the simplification claimed here is at
the level that determines the constant.  After normalization, the previous
token redistribution is replaced by one exact local signature lemma
(\cref{lem:signature}) and one circuit-rank identity
(\cref{lem:root-support}(iv)); the improvement from $9/16$ to $5/8$ then
reduces to a two-case degree count whose threshold is the crossing point of
two elementary rational functions.  A reader who accepts
\cref{lem:normalization} can follow the entire derivation of the constant in
about three pages.

We do not claim that the whole implementation is elementary.  Normalization
and the exact maximum-density subroutine remain technical polynomial-time
infrastructure, and they are the reason the paper is longer than its
mathematical core.  They are stated and proved here for completeness, not
because the constant depends on their details.

\subsection{Tightness and scope of the barrier}

For every $q\ge3$, Section~\ref{sec:barrier} constructs a normalized
half-integral point whose maximum projected density
$\dens_x(W)$---the LP mass inside $W$ divided by $|W|-1$, defined formally in
\cref{sec:projection}---equals
\[
  \max_{|W|\ge2}\dens_x(W)=\frac{5q}{2(4q-1)}
  \longrightarrow\frac58.
\]
Hence no universal theorem asserting $\max_W\dens_x(W)\ge\eta$ for every
normalized half-integral point can hold with $\eta>5/8$.  With unit costs the
whole support is one maximum-density set, so contracting it is a legal
execution of the framework, and its ratio against the displayed point is
$8/5-2/(5q)$.  In other words, there \emph{exists} an execution whose charge
approaches $8/5$.  We do not prove that every choice among maximum-density
sets is bad, and the construction leaves open whether some specified rule for
choosing among them does better.

\subsection{Organization}

Section~\ref{sec:model} defines Forest-BCR, normalization, projection, and
recursive contraction.  Section~\ref{sec:root-structure} develops the exact
degree-two signature and root-support identities.
Section~\ref{sec:cycle-rank} proves the circuit-rank inequality and the
$5/8$ density theorem.  \Cref{sec:barrier} shows that the density bound is
asymptotically tight.  \Cref{app:preprocessing} proves polynomial-time
normalization, \cref{app:inherited} reproduces the proofs of the three
framework lemmas inherited from~\cite{BGT26}, and
\cref{app:barrier-details} contains the remaining barrier
calculations.  A reader interested only in the constant may read
\cref{sec:balance-intro}, \cref{sec:root-structure}, and
\cref{sec:cycle-rank}.
\section{Forest-BCR, normalization, and contraction}\label{sec:model}

\subsection{The relaxation and terminology}

Let $G=(V,E,c)$ be a finite undirected multigraph with nonnegative rational
edge costs encoded in binary.  The demand list $\calP$ is finite, repetitions
are allowed, and every $P\in\calP$ is an unordered pair of distinct vertices.
A vertex is a \emph{terminal} if it is an endpoint of some demand; all other
vertices are \emph{Steiner vertices}.

We may assume that $G$ is connected and simple.  Feasibility forces the two
endpoints of every demand into one connected component of $G$: otherwise, for
any root receiving positive assignment, the component containing an endpoint
but not that root is a valid cut with zero outgoing capacity, contradicting
\eqref{eq:forest-bcr-cuts}.  We may therefore treat the demand-bearing
components independently and output the union of the returned forests.  Within
a component, parallel edges are merged by moving all arc mass of each root and
direction onto a cheapest parallel copy; every cut value is unchanged and the
cost does not increase.

We then pass to the shortest-path completion of the resulting connected
graph, which is a finite pseudometric because zero-cost edges are allowed.
Arc mass is carried to the completion arc on the same ordered pair, so every
cut value is unchanged and the cost cannot increase.  Each completion edge
stores one cheapest path in the preceding graph, so a purchased completion
edge can later be expanded at no extra cost.

During recursion we contract a vertex set to one quotient vertex and take the
shortest-path completion of the quotient.  At the end, purchased edges are
expanded one recursion level at a time, each edge being expanded once; the
result is a set of original edges, so the output has size polynomial in the
input.  An expanded path may enter a fiber contracted earlier, but the
spanning tree bought when that fiber was contracted already connects all of
its vertices, so this introduces no uncharged cost.

Let $\vec E$ contain both arcs $(u,v)$ and $(v,u)$ for every undirected
edge $e=\{u,v\}\in E$, and assign both arcs the edge cost:
$c_{(u,v)}=c_{(v,u)}:=c_e$.  Forest-BCR has assignment variables $z_P^r$ and
arc variables $x_a^r$, where $r\in V$ is a candidate root.  It is
\begin{alignat}{3}
  \min\quad &c(x):=\sum_{r\in V}\sum_{a\in\vec E}c_a x_a^r
  \label{eq:objective}\\[-.4em]
  \text{s.t.}\quad
  &\sum_{r\in V}z_P^r=1
  &&\qquad(P\in\calP),\label{eq:zsum}\\
  &x^r(\delta^+(S))\ge z_P^r
  &&\qquad\bigl(r\in V,\ P\in\calP,\ S\subseteq V\setminus\{r\},\
                    P\cap S\ne\varnothing\bigr),\label{eq:forest-bcr-cuts}\\
  &x,z\ge0.\notag
\end{alignat}
For an arc $a=(u,v)$ we abbreviate $x^r_{(u,v)}$ by $x^r_{uv}$, and
$c_{uv}$ denotes the metric cost of the undirected edge $\{u,v\}$.
We call a set $S$ appearing in \eqref{eq:forest-bcr-cuts}
\emph{valid for $(r,P)$}.  A root $r$ is \emph{participating} if $z_P^r>0$ for some
$P$, that is, if at least one demand sends positive mass to it.
A point is \emph{half-integral} if every coordinate belongs to
$\tfrac12\mathbb Z_{\ge0}$.

For every $z_P^r>0$ and every endpoint $v\in P\setminus\{r\}$, the cut
constraints are equivalent by max-flow/min-cut to an $x^r$-capacitated
$v$--$r$ flow of value $z_P^r$.  The value of Forest-BCR depends on the
chosen demand representation~\cite{BGT26}; all statements below concern the
fixed list $\calP$.

\subsection{Normalization}

Let $\varepsilon>0$.  For pairwise distinct $u,v,w$, the
\emph{$\varepsilon$-split at $(u,v,w)$ in root $r$} decreases
$x_{uv}^r$ and $x_{vw}^r$ by $\varepsilon$ and increases $x_{uw}^r$ by
$\varepsilon$.  It is \emph{feasible} if the resulting point is nonnegative
and satisfies \eqref{eq:forest-bcr-cuts}.  A point is \emph{split-free} if it admits no
feasible positive split, and it is \emph{fully reduced} if decreasing any
single positive arc coordinate by any positive amount destroys feasibility.
It is \emph{normalized} if it is fully reduced, split-free, and every participating
root is a terminal.  Directed splitting off is a special case of the classical
splitting-off operation; see Frank~\cite{Frank11} for general background.

\begin{lemma}[Polynomial-time normalization]\label{lem:normalization}
From every half-integral feasible point of a metric instance one can compute
in polynomial time a normalized half-integral feasible point of no larger
cost.
\end{lemma}

This normal form is exactly the one used by Byrka, Grandoni, and Traub: what
we call \emph{normalized} is what they call \emph{fully reduced}, namely a
\emph{well structured} point --- no assignment to a nonterminal root and no
feasible split --- that in addition admits no coordinatewise
decrease~\cite[Sec.~2.1 and the proof of Lem.~5]{BGT26}.  Their rerouting
step is~\cite[Lem.~2]{BGT26} and their splitting-off schedule is in the
proof of~\cite[Lem.~3]{BGT26}.  We claim no stronger normalization theorem than
theirs.  Because the whole analysis is carried out on normalized points, we
give a self-contained polynomial-time implementation in
\cref{app:preprocessing}: it combines root rerouting, coordinatewise reduction
by directed minimum cuts, and an exact maximum-feasible-split computation over
both cut-membership patterns in which a split can lower capacity.  Metricity
makes every shortcut no more expensive.

The next lemma packages full reduction in the form we use throughout.  The
flow reading of the cut constraints is already used
in~\cite[proof of Lem.~6]{BGT26}; the covering statement --- that every
positive arc lies on some chosen witness path --- is the part we rely on, and
it is a direct consequence of full reduction rather than a new ingredient.

\begin{lemma}[Path witnesses]\label{lem:witness-paths}
Let $(x,z)$ be fully reduced and let $r$ be a root.  For every triple
$(r,P,v)$ with $z_P^r>0$ and $v\in P\setminus\{r\}$, one can choose a
circulation-free $x^r$-capacitated $v$--$r$ flow of value $z_P^r$, together
with a decomposition into simple $v$--$r$ paths, such that every arc $a$ with
$x_a^r>0$ lies on at least one path chosen for one of these triples.
The capacity constraint is imposed on each witness flow separately; the
witness flows of different triples are not required to be simultaneously
routable, and their sum may well exceed $x^r$.
\end{lemma}

The idea is short.  Each cut constraint is, by max-flow/min-cut, a flow
requirement, and scaling capacities by two makes those flows integral, so a
path decomposition exists without leaving the half-integral world.  Full
reduction supplies the covering statement: an arc used by no witness could
have its coordinate lowered without disturbing any witness, which is exactly
what full reduction forbids.  The proof is in \cref{app:inherited}.

\subsection{Projection and density}\label{sec:projection}

For every root $r$ and every arc $a=(u,v)$, replace the half-integral
coordinate $x^r_a$ by $2x^r_a$ parallel \emph{half-unit copies} of the
undirected edge $uv$, each oriented from $u$ to $v$ and labeled $r$.  The
multiplicity is an integer because the point is half-integral.  After deleting
isolated vertices, the resulting labeled multigraph is the \emph{projection}
$\wideG=(V(\wideG),\wideE)$, and its vertices are the \emph{support
vertices}.  This convention gives every projected edge an orientation and a
root label even when a single coordinate exceeds $1/2$.  Counting
multiplicity, every support vertex satisfies
\begin{equation}\label{eq:projection-degree}
  d_{\wideG}(v)=2\sum_{r}\bigl(x^r(\delta^+(v))+x^r(\delta^-(v))\bigr),
\end{equation}
so the projected degree is exactly twice the incident LP mass.  We reserve
\emph{root support} for the graphs $D_r$ defined in
\cref{sec:root-structure}.

For $W\subseteq V(\wideG)$ with $|W|\ge2$, put
\[
  x(E[W]):=\sum_{uv\in E[W]}\sum_r(x_{uv}^r+x_{vu}^r),
  \qquad
  \dens_x(W):=\frac{x(E[W])}{|W|-1}
             =\frac{|\wideE[W]|}{2(|W|-1)}.
\]
The denominator is the number of edges a spanning tree of $W$ must buy and the
numerator is the LP mass available to pay for them, so $\dens_x$ is
\emph{LP mass per tree edge}.  In classical terminology,
\[
  2\max_{|W|\ge2}\dens_x(W)
  =\max_{|W|\ge2}\frac{|\wideE[W]|}{|W|-1}
\]
is the fractional arboricity of the projection~\cite{NashWilliams64}.  Two consequences are used repeatedly: a set of
density $\ge\eta$ yields a tree costing at most $1/\eta$ times the enclosed LP
cost (\cref{lem:contraction}); and if $W=V(D)$ for a connected $D$, then
$|W|-1=|E(D)|-\beta(D)$, so \emph{each unit of circuit rank removes one edge
from the bill}.

\begin{remark}[Exact deterministic maximum density]
\label{rem:densest-computation}
An exact maximizer of $\dens_x$ is computable in polynomial time by the
method of Byrka, Grandoni, and Traub~\cite[Lem.~1]{BGT26}.  For a rational value $\gamma$, define
\[
  F_\gamma(W):=\gamma|W|-x(E[W]).
\]
The internal-edge function $x(E[W])$ is supermodular, and hence
$F_\gamma$ is submodular.  After forcing a pair of vertices into $W$, the
inequality $\dens_x(W)\ge\gamma$ is equivalent to
$F_\gamma(W)\le\gamma$; the minimum can therefore be found by polynomial-time
submodular-function minimization~\cite{Schrijver00}.  The possible densities
are among the $O(n|\wideE|)$ ratios $k/(2\ell)$ with
$0\le k\le|\wideE|$ and $1\le\ell\le n-1$, so exact search over this grid is
polynomial in the explicit input length.

Scanning the candidate values in decreasing order, and for each value the
forced pairs in a fixed vertex order, the procedure returns one specific
maximum-density set; the algorithm below always uses that set.  The
approximation analysis is valid for \emph{every} maximum-density choice, so
this convention only serves to make the algorithm deterministic.
\end{remark}

\subsection{Densest-set contraction}

The algorithm below is~\cite[Alg.~1]{BGT26} and the following lemma is
assembled from~\cite[Lem.~4]{BGT26} and the proof
of~\cite[Thm.~1]{BGT26}.  Nothing in this subsection is new.  We restate it
because our contribution is precisely to improve the constant $\eta$ with
which it may be invoked, and because we need the normalization step and the
degenerate cases to be explicit in the induction.

\medskip
\noindent\textbf{Procedure $\Round(G,c,\calP,x,z)$.}
\begin{enumerate}[label=\arabic*.,leftmargin=2.1em,itemsep=.15em,topsep=.25em]
\item If $\calP=\varnothing$, return the empty forest.
\item Normalize $(x,z)$ using \cref{lem:normalization}.
\item Let $W$ be the maximum-density set returned by \cref{rem:densest-computation}.
\item Buy a minimum spanning tree on $W$.
\item Contract $W$, delete demands made internal, and recurse; on return,
      expand all predecessor paths, take the union, and delete cycles.
\end{enumerate}

\begin{lemma}[Densest-set contraction]\label{lem:contraction}
Let $\eta>0$.  Suppose every normalized half-integral feasible point of an
instance with at least one demand admits a set $W$, $|W|\ge2$, with
$\dens_x(W)\ge\eta$.  Then $\Round$ returns a feasible Steiner forest of cost at most
$c(x)/\eta$.
\end{lemma}

Only the reason it works matters for the sequel, and it fits in one paragraph.
Write $\varrho:=\dens_x(W)\ge\eta$ and divide the LP mass inside $W$ by
$\varrho$.  The resulting vector has total weight exactly $|W|-1$, and because
$W$ has \emph{maximum} density, no $Q\subseteq W$ receives more than $|Q|-1$;
these are precisely Edmonds' spanning-tree polytope inequalities, so a minimum
spanning tree on $W$ costs at most $1/\varrho\le1/\eta$ times the LP mass
buried inside $W$.  That mass then disappears from the recursive bill, because
contracting $W$ deletes exactly the arcs internal to it.  Contraction also
preserves feasibility, since a valid cut in the quotient pulls back to a valid
cut upstairs.  Induction on the number of vertices gives the factor $1/\eta$
globally.  The full argument, including quotient half-integrality, the
expansion of contracted paths, and the final cycle deletion, is reproduced in
\cref{app:inherited}.

\subsection{Degree normal form}

\begin{lemma}[Projection degree normal form]\label{lem:degree-normal-form}
Let $(x,z)$ be normalized.  Every vertex of $\wideG$ has degree at least two,
and every Steiner vertex of $\wideG$ has degree at least three.
\end{lemma}

This is~\cite[Lem.~6 and Claim~1 of its proof]{BGT26}, and it is what makes
``degree two'' the interesting case.  Witness paths must enter and leave a Steiner
vertex, which gives degree at least two; a Steiner vertex of degree exactly
two carries one unit of mass on two half-arcs that a path traverses in
sequence, and splitting them off is then feasible, contradicting
split-freeness.  Terminals have degree at least two because the singleton and
complementary-singleton cuts already demand one full unit of incident mass.
The proof is in \cref{app:inherited}.

A vertex of projected degree two is called \emph{low}; a support vertex of
degree at least three is \emph{high}.

\begin{remark}[Simple-projection dictionary]\label{rem:simple-dictionary}
By \cref{lem:degree-normal-form}, low vertices are exactly the support
vertices of minimum possible degree.  If $\wideG$ has two parallel edges
between $u$ and $v$, then $\dens_x(\{u,v\})\ge1$.  Hence only simple projections are difficult.  In a
simple projection of a normalized point, every low vertex is a terminal, every
positive coordinate equals $1/2$, and each projected edge has a unique
orientation and root label.
\end{remark}
\section{Degree-two signatures and root supports}\label{sec:root-structure}

\paragraph{Standing assumption.}
Throughout \cref{sec:root-structure,sec:cycle-rank} we assume that the
projection $\wideG$ is \emph{simple}.  This is without loss of generality for
\cref{thm:density}: by \cref{rem:simple-dictionary}, a pair of parallel
projected edges already exhibits a two-vertex set of density at least one.
Under this assumption every positive coordinate equals $1/2$, and every
projected edge carries exactly one orientation and one root label; all objects
introduced below are defined under this convention.

The first new ingredient is that a low terminal determines the complete
root-assignment vector of all its incident demands.  The outcome is the
dichotomy of \cref{lem:low-types}; \cref{sec:cycle-rank} converts it into
$\beta_C\ge L_C/2$ and then, by a two-case degree count, into the universal
$5/8$ bound.

\begin{lemma}[Low-terminal signature]\label{lem:signature}
Let $v$ be a low terminal and define
\[
  \sigma_v(r):=
  \begin{cases}
    x^r(\delta^+(\{v\})),&r\ne v,\\
    x^v(\delta^-(\{v\})),&r=v.
  \end{cases}
\]
For every demand $P$ incident with $v$ and every root $r$,
\[
  z_P^r=\sigma_v(r).
\]
\end{lemma}

\begin{proof}
For $r\ne v$, apply \eqref{eq:forest-bcr-cuts} to $S=\{v\}$.  For $r=v$,
apply it to $S=V\setminus\{v\}$; the other endpoint of $P$ lies in this set.
Thus $\sigma_v(r)\ge z_P^r$ for every $r$.  The arc sets counted by the
coordinates of $\sigma_v$ are pairwise disjoint and lie within the total
incident mass, which is one because $v$ has projected degree two.  Therefore
\[
  1\ge\sum_r\sigma_v(r)\ge\sum_r z_P^r=1,
\]
so every coordinate inequality is tight.
\end{proof}

\begin{corollary}[No freedom at degree two]\label{cor:signature-orientation}
All demands incident with a low terminal $v$ have the same assignment vector.
This vector is either a unit vector $\mathbf e_r$ or has exactly two
nonzero entries, at roots $r$ and $s$, each equal to one half.  Moreover:
\begin{enumerate}[label=(\roman*),leftmargin=1.8em]
\item if the vector is $\mathbf e_r$ and $v\ne r$, both incident half-arcs
      leave $v$ and are labeled $r$;
\item if the vector is $\mathbf e_v$, both incident half-arcs enter $v$ and
      are labeled $v$;
\item in the split case, one incident half-arc is labeled $r$ and one is
      labeled $s$; the arc of label $v$, if $v\in\{r,s\}$, enters $v$, and
      every other one leaves $v$.
\end{enumerate}
\end{corollary}

\begin{proof}
Half-integrality and \eqref{eq:zsum} give the two possible vectors.  The proof
of \cref{lem:signature} accounts for the entire unit of incident mass, so no
other label or orientation is possible.
\end{proof}

For a participating root $r$, let $D_r$ be the oriented graph consisting of the
projected edges labeled $r$, their endpoints, and the root $r$.

\begin{lemma}[Rootward minimal supports]\label{lem:root-support}
For every participating root $r$:
\begin{enumerate}[label=(\roman*),leftmargin=1.8em]
  \item $D_r$ is connected;
  \item no arc leaves $r$;
  \item every other vertex of $D_r$ has outdegree at least one; and
  \item
  \[
    \beta(D_r):=|E(D_r)|-|V(D_r)|+1
    =\sum_{v\in V(D_r)\setminus\{r\}}\bigl(d^+_{D_r}(v)-1\bigr).
  \]
\end{enumerate}
Every projected edge is labeled by a participating root.
\end{lemma}

\begin{proof}
By \cref{lem:witness-paths}, every positive root-$r$ arc lies on a simple
path ending at $r$.  The union of those paths is connected, no arc leaves the
root, and every nonroot support vertex has an outgoing arc.  Summing
outdegrees counts the edges, and subtracting one for every nonroot vertex
gives (iv).  Informally, (iv) says that each nonroot vertex must spend one
outgoing arc merely to make progress toward $r$---together a spanning tree's
worth---so that \emph{circuit rank equals total excess outdegree}.  We call
such a support \emph{rootward}: every nonroot vertex has a directed path to
$r$ and no arc leaves $r$.  In particular, $\beta(D_r)=0$ exactly when the
rootward support is an arborescence directed toward $r$.  Finally, if $x^r\ne0$ but all assignments $z_P^r$ were zero, the
entire root layer could be deleted, contradicting full reduction.
\end{proof}

A demand is \emph{unit-assigned} if $z_P^r=1$ for one root, and
\emph{split-assigned} if it has value one half at two roots.  These are the
only possibilities.

\begin{lemma}[Unit assignments create cycles]\label{lem:unit-cycle}
If $z_P^r=1$, then $\beta(D_r)\ge1$.
\end{lemma}

\begin{proof}
Choose $v\in P\setminus\{r\}$.  Scale the root-$r$ capacities by two.  Every
$v$--$r$ cut has integral capacity at least two, while
\cref{rem:simple-dictionary} makes every positive scaled arc have capacity
one.  Compute an integral maximum flow, retain an exact two-unit subflow from
its path decomposition, and discard circulations.  The result decomposes into
two arc-disjoint $v$--$r$ paths.  A simple
projection has at most one positive directed arc on each undirected pair, so
the paths are also edge-disjoint in the undirected support.  They join the
same distinct vertices, so their nonempty union contains a cycle in $D_r$.
\end{proof}

Define the \emph{split-root graph} $J$ as the multigraph whose vertices are
the participating roots and in which every split demand assigned to $r,s$ gives one
edge $rs$.  Call a root \emph{unit-bearing} if some demand is assigned to it with value one; it may also participate in split assignments.  For a
connected component $C$ of $J$, set
\[
  D_C:=\bigcup_{r\in C}D_r,
  \qquad m_C:=|E(D_C)|,
  \qquad \beta_C:=m_C-|V(D_C)|+1.
\]
The union is connected: if an edge $rs$ of $J$ comes from a split demand $P$,
both endpoints of $P$ lie in $V(D_r)\cap V(D_s)$; connectivity of $J[C]$
then gives connectivity of $D_C$.

\begin{lemma}[The two low-vertex types]\label{lem:low-types}
Let $v\in V(D_C)$ be globally low in the full projection $\wideG$.  Exactly
one of the following holds; see \cref{fig:low-types}.
\begin{enumerate}[label=(\roman*),leftmargin=1.8em]
  \item \emph{Unit-low.}  Every demand incident with $v$ is assigned
  integrally to one root $r\in C$.  Both incident edges have label $r$, and
  $v$ belongs to no other root support.  If $v\ne r$, both arcs leave $v$.

  \item \emph{Split-low.}  Every demand incident with $v$ is split between
  two roots $r,s\in C$.  One incident edge has label $r$, the other label
  $s$, and $v$ belongs exactly to $D_r$ and $D_s$.  Moreover $rs\in E(J)$.
\end{enumerate}
Every low vertex belongs to a unique component of $J$.
\end{lemma}

\begin{proof}
The assignment alternatives and the incident labels follow from
\cref{cor:signature-orientation}.  In the unit case, membership in a support
with another label would require another incident projected edge.  Suppose
$v$ were itself a participating root $s$ distinct from the labels already present.
Some demand $P'$ has $z_{P'}^s>0$, and an endpoint
$p\in P'\setminus\{s\}$ exists because demand endpoints are distinct.
Feasibility gives a positive $x^s$-flow from $p$ into $s=v$, hence an
$s$-labeled arc incident with $v$, contradicting the two existing labels.
Thus $v$ belongs only to the stated support.

In the split case, the two incident edges already account for the entire
incident mass, and the same participating-root argument excludes a third support.
Any demand incident with $v$ gives the edge $rs$ of $J$.  The final statement
follows because a unit-bearing root lies in one component of $J$, while the two roots
of a split assignment are adjacent in $J$.
\end{proof}

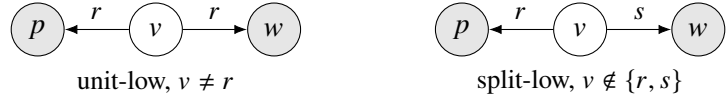
\begin{figure}[t]
\centering
\begin{tikzpicture}[>=Latex,scale=.92,
  low/.style={circle,draw,inner sep=1.5pt,minimum size=7mm},
  high/.style={circle,draw,fill=black!10,inner sep=1.5pt,minimum size=7mm},
  lab/.style={font=\small}]
  \node[high] (a) at (0,0) {$p$};
  \node[low]  (v) at (1.7,0) {$v$};
  \node[high] (b) at (3.4,0) {$w$};
  \draw[->] (v) -- node[above,lab] {$r$} (a);
  \draw[->] (v) -- node[above,lab] {$r$} (b);
  \node[lab] at (1.7,-.75) {unit-low, $v\ne r$};

  \node[high] (c) at (6.1,0) {$p$};
  \node[low]  (w) at (7.8,0) {$v$};
  \node[high] (d) at (9.5,0) {$w$};
  \draw[->] (w) -- node[above,lab] {$r$} (c);
  \draw[->] (w) -- node[above,lab] {$s$} (d);
  \node[lab] at (7.8,-.75) {split-low, $v\notin\{r,s\}$};
\end{tikzpicture}
\caption{A nonroot unit-low vertex contributes one to the circuit-rank
identity of one support.  A split-low vertex contributes one to the overlap
of two supports.  If the low vertex equals a root, the edge of that root
label is directed into it.}
\label{fig:low-types}
\end{figure}

As a running example, \cref{fig:running} shows the family of
\cref{sec:barrier} at $q=3$: every $t_i$ and every $r_i$ is split-low with
roots $r_i,r_{i+1}$, there are no unit-low vertices, the split-root graph is a
$q$-cycle, and its single component has $L_C=2q$ and $\beta_C=q+1$.  
It also gives a compact way to read \cref{lem:low-types} off a picture: color
each projected edge by the root it serves.  In this particular example
\emph{the low vertices are exactly the vertices at which the color changes},
while high vertices lie in the interior of a single color class.  In general
only one half of that slogan survives: a split-low vertex is always a color
change, whereas a unit-low vertex sees a single color, and a high vertex may
well lie in several supports.

\begin{figure}[t]
\centering
\begin{tikzpicture}[>=Latex,scale=.80,
  cir/.style={circle,draw,thick,inner sep=0pt,minimum size=6.6mm,
              font=\small,fill=white},
  low/.style={cir,line width=1.05pt},
  high/.style={cir,fill=black!7,draw=Slate},
  e/.style={-{Latex[length=2.6mm,width=2.1mm]},line width=1.15pt,
            shorten >=1.1mm,shorten <=1.1mm},
  eA/.style={e,RootA},
  eB/.style={e,RootB,densely dashed},
  eC/.style={e,RootC,densely dotted,line width=1.5pt}]
\node[high] (a0) at (-3.250,0.573) {$a_0$};
\node[high] (b0) at (-2.121,2.528) {$b_0$};
\node[high] (a1) at (2.121,2.528)  {$a_1$};
\node[high] (b1) at (3.250,0.573)  {$b_1$};
\node[high] (a2) at (1.129,-3.101) {$a_2$};
\node[high] (b2) at (-1.129,-3.101){$b_2$};
\node[low]  (r0) at (0,3.3)        {$r_0$};
\node[low]  (r1) at (2.858,-1.65)  {$r_1$};
\node[low]  (r2) at (-2.858,-1.65) {$r_2$};
\node[low]  (t0) at (0,1.42)       {$t_0$};
\node[low]  (t1) at (1.230,-0.71)  {$t_1$};
\node[low]  (t2) at (-1.230,-0.71) {$t_2$};
\draw[eA] (r2)--(a0); \draw[eA] (a0)--(b0);
\draw[eA] (b0)--(r0); \draw[eA] (t0)--(a0);
\draw[eA] (t2)--(b0);
\draw[eB] (r0)--(a1); \draw[eB] (a1)--(b1);
\draw[eB] (b1)--(r1); \draw[eB] (t1)--(a1);
\draw[eB] (t0)--(b1);
\draw[eC] (r1)--(a2); \draw[eC] (a2)--(b2);
\draw[eC] (b2)--(r2); \draw[eC] (t2)--(a2);
\draw[eC] (t1)--(b2);
\draw[eA] (3.95,2.75)--(4.75,2.75);
\node[font=\small,anchor=west] at (4.85,2.75) {arcs of $D_{r_0}$};
\draw[eB] (3.95,2.30)--(4.75,2.30);
\node[font=\small,anchor=west] at (4.85,2.30) {arcs of $D_{r_1}$};
\draw[eC] (3.95,1.85)--(4.75,1.85);
\node[font=\small,anchor=west] at (4.85,1.85) {arcs of $D_{r_2}$};
\node[font=\small,align=left,anchor=west,text width=45mm] at (3.9,0.6)
  {Demands $P_i=\{r_i,t_i\}$, split one half each between $r_i$ and
   $r_{i+1}$; every arc carries $\tfrac12$.};
\node[font=\small,align=left,anchor=west,text width=45mm] at (3.9,-1.9)
  {Bold circles: degree $2$ (low).  Gray: degree $3$ (high).
   Here $n=12$, $m=15$, $\beta=4$, $L/m=\tfrac25$.};
\end{tikzpicture}
\caption{The family of \cref{sec:barrier} at $q=3$, used as a running example.
Each color class is a rootward arborescence pointing at its own root; high
vertices are monochromatic; in this example the low vertices are exactly the
color changes.  The arrows show orientations and root labels; the underlying
undirected projection is a $9$-cycle with three length-two handles, so
$\beta=1+3=4$.}
\label{fig:running}
\end{figure}
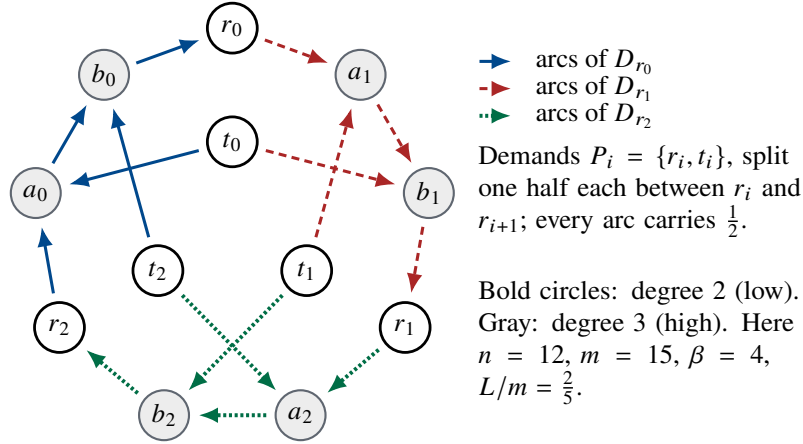

\section{Circuit rank and the \texorpdfstring{$5/8$}{5/8} density theorem}\label{sec:cycle-rank}

We now turn the local dichotomy into a global count.  The proof has two
ledgers: circuit rank internal to one root support, and vertex overlap among
different supports.

Define the \emph{overlap surplus} of a component $C$ by
\[
  \iota_C:=\sum_{r\in C}|V(D_r)|-|V(D_C)|.
\]
Equivalently, $\iota_C$ counts each vertex of $D_C$ once for every root
support beyond the first that contains it; in particular every summand of that
count is nonnegative.

It is convenient to fix, once and for all, the counters attached to a
component $C$ of $J$:
\begin{equation}\label{eq:component-counters}
\begin{aligned}
  a_C&:=\#\{\text{unit-bearing roots in }C\}, &\qquad
  u_C&:=\#\{\text{unit-low vertices of }D_C\},\\[1pt]
  L_C&:=\#\{\text{low vertices of }D_C\}, &\qquad
  s_C&:=\#\{\text{split-low vertices of }D_C\},
\end{aligned}
\end{equation}
where ``low'' always means low in the full projection $\wideG$.  By
\cref{lem:low-types} every low vertex of $D_C$ is of exactly one of the two
types, so $L_C=u_C+s_C$.  All four counters refer to the component $C$ only;
the number of root supports in $C$ is written $|C|$ throughout.
\begin{figure}[t]
\centering
\begin{tikzpicture}[scale=.88,font=\small,
  bl/.style={ellipse,draw,thick,minimum width=28mm,minimum height=15mm},
  sh/.style={circle,fill=RootB,inner sep=1.5pt}]
\node[bl,draw=RootA,fill=TintA] (A) at (0,0) {};
\node[bl,draw=RootC,fill=TintC] (B) at (2.2,0) {};
\node[sh] at (1.1,0) {};
\node at (1.1,-1.35) {one shared vertex};
\node[RootA] at (1.1,-1.85) {$0$ new cycles};
\node[bl,draw=RootA,fill=TintA] (C) at (7.8,0) {};
\node[bl,draw=RootC,fill=TintC] (D) at (10.0,0) {};
\node[sh] at (8.9,.40) {}; \node[sh] at (8.9,-.40) {};
\node at (8.9,-1.35) {two shared vertices};
\node[RootB] at (8.9,-1.85) {$1$ new cycle};
\draw[->,thick,RootB] (8.9,.40) to[bend right=55] (8.9,-.40);
\draw[->,thick,RootB] (8.9,-.40) to[bend right=55] (8.9,.40);
\end{tikzpicture}
\caption{The union of two connected, \emph{edge-disjoint} graphs whose vertex
sets meet in exactly $t\ge1$ vertices has circuit rank equal to the sum of
their circuit ranks plus $t-1$: one shared vertex is spent on connectivity,
and every further one contributes an independent cycle.  Edge-disjointness is
what the root labels supply, since each projected edge carries one label.
This is the content of \eqref{eq:overlap-identity}, where $|C|-1$ of the
$\iota_C$ shared vertices are consumed to connect the $|C|$ supports.}
\label{fig:gluing}
\end{figure}
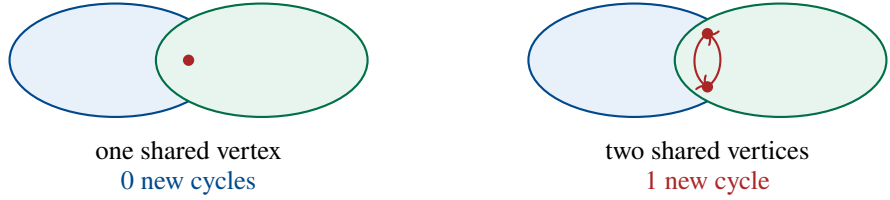

The two supports glued in \cref{fig:gluing} illustrate the mechanism.
Because root labels partition projected edges, the connected root supports
have pairwise disjoint edge sets, and
\begin{equation}\label{eq:overlap-identity}
  \beta_C
  =\sum_{r\in C}\beta(D_r)+\iota_C-(|C|-1).
\end{equation}
Indeed, substituting the definition of $\iota_C$,
$m_C=\sum_{r\in C}|E(D_r)|$ (labels partition the edges), and
$\beta(D_r)=|E(D_r)|-|V(D_r)|+1$ (each support is connected by
\cref{lem:root-support}) into both sides proves \eqref{eq:overlap-identity}.

\begin{lemma}[Low vertices force half a cycle]\label{lem:half-cycle}
For every component $C$ of the split-root graph,
\[
  \beta_C\ge\frac{L_C}{2}.
\]
\end{lemma}

\begin{proof}
We use the counters \eqref{eq:component-counters} and establish four
simultaneous inequalities.

\smallskip
\noindent\textbf{Claim 1:}
$\sum_{r\in C}\beta(D_r)\ge a_C$.
Every unit-bearing root has circuit rank at least one by
\cref{lem:unit-cycle}.

\smallskip
\noindent\textbf{Claim 2:}
$\sum_{r\in C}\beta(D_r)\ge u_C-a_C$.
A unit-low vertex lies in the support of a unit-bearing root.  At most one such
vertex per unit-bearing root can coincide with that root; every other unit-low vertex has
outdegree two and contributes one to the nonnegative sum in
\cref{lem:root-support}(iv).

\smallskip
\noindent\textbf{Claim 3:}
$\iota_C\ge s_C$.
Every split-low vertex lies in exactly two supports, while every term
in the definition of $\iota_C$ is nonnegative.

\smallskip
\noindent\textbf{Claim 4:}
$\iota_C\ge2(|C|-1)$.
Add the supports along a spanning tree of $J[C]$.  Whenever a new support is
added, the two distinct endpoints of the split demand corresponding to its
parent tree edge lie both in the new support and in the current union
(as observed before \cref{lem:low-types}, each endpoint of a split demand
lies in both root supports).  Thus
its addition increases the union's vertex count by at most
$|V(D_r)|-2$.  Equivalently,
$\sum_r|V(D_r)|-|V(D_C)|\ge2(|C|-1)$.

Claims 1 and 4 in \eqref{eq:overlap-identity} give
$\beta_C\ge a_C+|C|-1$; Claims 2 and 3 give
$\beta_C\ge L_C-a_C-(|C|-1)$.  The maximum of these two lower bounds is at least
their average, $L_C/2$.
\end{proof}

\begin{theorem}[Universal $5/8$ density]\label{thm:density}
Every normalized half-integral Forest-BCR point for an instance with at least
one demand admits a set $W$, $|W|\ge2$, such that
\[
  \dens_x(W)\ge\frac58.
\]
\end{theorem}

This is the only statement of the analysis that we change.  It has exactly the
hypotheses and the shape of~\cite[Lem.~5]{BGT26}, which gives $9/16$ in
place of $5/8$; substituting it into \cref{lem:contraction} replaces $16/9$
by $8/5$ and leaves every other step of~\cite{BGT26} untouched.

\begin{proof}
A parallel projected pair gives density at least one, so assume the projection
is simple.  Let $m:=|\wideE|$, $n:=|V(\wideG)|$, and let $L$ be the total
number of low vertices.  We run the computation of
\cref{sec:balance-intro} with a general threshold $\theta\in(0,1)$ and only at
the end set $\theta=2/5$.

\smallskip
\noindent\emph{Case 1: some component $C$ has $L_C\ge\theta m_C$.}
Take $W=V(D_C)$, which has $|W|\ge2$ because $D_C$ has an edge.  The induced
projection contains all $m_C$ edges of the connected graph $D_C$ (possibly
more, if edges labeled by roots outside $C$ also have both endpoints in $W$;
this only helps), and $|W|-1=m_C-\beta_C$.  Using \cref{lem:half-cycle},
\[
  \dens_x(W)
  \ \ge\ \frac{m_C}{2(m_C-\beta_C)}
  \ \ge\ \frac{m_C}{2(m_C-L_C/2)}
  \ \ge\ \frac{m_C}{2\bigl(m_C-\tfrac{\theta m_C}{2}\bigr)}
  \ =\ \frac{1}{2-\theta},
\]
which is \eqref{eq:intro-local}.

\smallskip
\noindent\emph{Case 2: $L_C<\theta m_C$ for every component.}
Root labels partition the projected edges
(\cref{lem:root-support}), and \cref{lem:low-types} says that the components
of $J$ partition the low vertices; hence $\sum_C m_C=m$, $\sum_C L_C=L$ and
therefore $L<\theta m$.  Since low vertices have degree two and all other
support vertices at least three, the degree sum gives
\[
  2m\ \ge\ 2L+3(n-L)\ =\ 3n-L\ >\ 3n-\theta m,
\]
so $n<(2+\theta)m/3$ and, taking $W=V(\wideG)$ (again $|W|\ge2$),
\[
  \dens_x(V(\wideG))
  \ =\ \frac{m}{2(n-1)}
  \ >\ \frac{m}{2n}
  \ >\ \frac{3}{2(2+\theta)}.
\]

\smallskip
\noindent
The first bound increases in $\theta$ and the second decreases, so the
guarantee $\min\bigl\{\tfrac1{2-\theta},\tfrac3{2(2+\theta)}\bigr\}$ is
maximized where they agree.  By \eqref{eq:crossing} that is $\theta=2/5$, with
common value $5/8$.
For every finite point, Case~2 is strict.  Consequently the infimum, over
normalized half-integral points with at least one demand, of the maximum
projected density is at least $5/8$; \cref{thm:barrier} shows that it equals
$5/8$.
\end{proof}

\begin{proof}[Proof of \cref{thm:main}]
If $c(x)=0$, every metric edge carrying positive LP mass has cost zero.  The
same procedure buys only zero-cost edges and returns a zero-cost forest, so
assume $c(x)>0$.

By \cref{sec:model} we may treat one demand-bearing connected component at a
time and take the union of the returned forests, so assume $G$ connected.
Pass to the shortest-path completion, which preserves feasibility and does not
increase $c(x)$, and run \Round, normalizing at every nonempty call.
By \cref{thm:density} every call has a set of density at least $5/8$, so
\cref{lem:contraction} with $\eta=5/8$ gives $c(F)\le\frac85c(x)$.  Expanding
the purchased completion edges does not increase the cost.

It remains to observe that the running time is polynomial.  Every subroutine
used is polynomial-time: shortest paths and quotient metrics; directed
minimum cuts and maximum flows, for normalization; and submodular-function
minimization, for the density decision.  Each is called polynomially often.
Indeed, \cref{app:preprocessing} performs polynomially many successful splits
and polynomially many coordinate reductions and reroutings, each of which
requires polynomially many min-cut or max-flow computations; one density
computation ranges over $O(n^2)$ forced pairs and a candidate grid of
$O(n|\wideE|)$ values; and every recursive call contracts at least two
vertices, so there are at most $n-1$ nonempty calls, while the current numbers
of vertices and demand occurrences never exceed their initial values.  All
numbers occurring are rationals of bit length polynomial in the input length
and the number of operations performed.  We make no attempt to optimize the
exponent.
\end{proof}

\section{The projected-density bound is asymptotically tight}\label{sec:barrier}

The constant $5/8$ in \cref{thm:density} cannot be replaced by any larger
constant.  All indices in this section are modulo $q$.

\subsection{A heuristic from approximate equality conditions}
\label{sec:derive}

The construction can be motivated by asking how the two cases of
\cref{thm:density} could be nearly tight simultaneously.  One expects roughly
$L=2m/5$, globally high vertices of degree three, few or no unit assignments,
and split-root supports that contribute almost exactly half a circuit per low
vertex.  The smallest symmetric rootward support with these features has five
half-arcs: three source streams merge through two degree-three Steiner
vertices into one root.

A tree-shaped split-root blueprint has a leaf defect.  Its \emph{whole-set}
density is below $5/8$, so \cref{thm:density} would force a denser proper
subset; this does not make the blueprint impossible, but it prevents the
whole projection from being extremal.  Closing the root interactions into a
cycle removes the leaves and yields the family below.  This discussion is
only a design heuristic.  Feasibility, normalization, and exact maximum
density are proved independently.

\subsection{The construction and its verification}

\begin{lemma}[Equal-marginal density certificate]
\label{lem:equal-marginal-certificate}
Let $H=(U,F)$ be connected.  If a distribution on spanning trees of $H$
includes every edge with probability at least $\rho$, then
\[
  \frac{|F[W]|}{|W|-1}\le\frac1\rho
  \qquad (W\subseteq U,\ |W|\ge2).
\]
If all edge marginals equal $\rho$, the whole vertex set attains equality.
\end{lemma}

\begin{proof}
For every $W$ and every spanning tree $T$,
$|E(T)\cap F[W]|\le |W|-1$.  Summing edge marginals gives
$\rho|F[W]|\le\mathbb E[|E(T)\cap F[W]|]\le|W|-1$.  For $W=U$, equality follows from
$|E(T)|=|U|-1$.
\end{proof}

\begin{theorem}[Exact density barrier]\label{thm:barrier}
For every integer $q\ge3$, there is an instance whose costs are the
shortest-path metric of a unit-cost graph, and a normalized half-integral
Forest-BCR point whose projection is simple with $4q$ vertices and $5q$ edges
and satisfies
\[
  \max_{|W|\ge2}\dens_x(W)=\frac{5q}{2(4q-1)}.
\]
In particular, the maximum projected density tends to $5/8$.
\end{theorem}

\begin{proof}
\emph{Construction.}
For each $i\in\mathbb Z_q$, introduce terminals $r_i,t_i$, Steiner vertices
$a_i,b_i$, and demand $P_i=\{r_i,t_i\}$.  Assign $P_i$ with value one half to
roots $r_i$ and $r_{i+1}$.  The root-$r_i$ support consists of the five
half-arcs
\begin{equation}\label{eq:barrier-arcs}
  t_i\to a_i,
  \quad r_{i-1}\to a_i,
  \quad a_i\to b_i,
  \quad t_{i-1}\to b_i,
  \quad b_i\to r_i.
\end{equation}
All five displayed arcs carry value one half.  By the flow criterion in
\cref{sec:model}, feasibility for root $r_i$ follows from the half-unit paths
$t_i\to a_i\to b_i\to r_i$ for $P_i$, and
$r_{i-1}\to a_i\to b_i\to r_i$ and
$t_{i-1}\to b_i\to r_i$ for $P_{i-1}$.

\emph{Normalization.}
Every displayed arc is the unique outgoing root-$r_i$ arc of a valid set whose
capacity therefore equals the required one half, so no coordinate can be
decreased and the point is fully reduced; the five certificates are listed in
\cref{tab:tight-arcs}.  A split needs a vertex with an incoming and an
outgoing arc of the same label, and only $a_i$ and $b_i$ have one, so there
are four candidate splits; each is blocked by a tight set of the first family
of \eqref{eq:two-negative-families}, as listed in
\cref{tab:blocking-cuts}.  Finally every participating root $r_i$ is an endpoint of
the demand $P_i$, hence a terminal.

\emph{Counting.}
The projection has the five edge types
\[
  t_i a_i,
  \quad r_{i-1}a_i,
  \quad a_i b_i,
  \quad t_{i-1}b_i,
  \quad b_i r_i.
\]
For $q\ge3$ these are distinct, so $n=4q$ and $m=5q$.

\emph{Maximum density.}
In \cref{app:barrier-details} we construct a distribution on spanning trees of
the projection in which every edge has inclusion probability
$\rho=(4q-1)/(5q)$.  The construction rests on one combinatorial observation,
proved there: deleting all edges of one of the four ``side'' classes leaves a
connected unicyclic graph whose unique cycle has length $3q$, whereas deleting
all edges ${a_i b_i}$ leaves $\gcd(q,2)$ disjoint cycles, because the walk
\[
  r_i\to a_{i+1}\to t_{i+1}\to b_{i+2}\to r_{i+2}
\]
advances the index by two.  Given the distribution,
\cref{lem:equal-marginal-certificate} gives $\dens_x(W)\le1/(2\rho)$ for every
$W$ with $|W|\ge2$, and the whole vertex set attains equality.
\end{proof}

For unit costs on the displayed support, the exhibited point has
$c(x)=5q/2$, while a minimum spanning tree on the whole vertex set costs
exactly $4q-1$.  Indeed the $5q$ support edges contain a spanning tree of
$4q-1$ unit edges, which gives the upper bound; conversely, in the
shortest-path metric of a unit-cost simple graph any two distinct vertices are
at distance at least one, so every spanning tree on $4q$ vertices costs at
least $4q-1$.  Hence the all-support contraction has ratio
\[
  \frac{4q-1}{5q/2}=\frac85-\frac{2}{5q}.
\]
The whole vertex set is a maximum-density set, so choosing it is a legal run
of the maximum-density/MST framework.  The family therefore makes the
universal $5/8$ density lemma tight and yields an $8/5-o(1)$ run under
adversarial tie-breaking, in the sense that one of the maximum-density sets
gives that ratio.  It does not show that every
maximum-density set is bad, and it leaves open whether some specified rule for
choosing among maximum-density sets has a better guarantee.

\begin{remark}[Why this is not an integrality-gap construction]
\label{rem:family-root-cover}
For even $q$, buy the five-edge supports of one side of the bipartition of the
split-root cycle.  Every demand $P_i$ is split between $r_i$ and $r_{i+1}$, so
one of its two roots is bought; and if $z_P^r>0$ then every endpoint of $P$
other than $r$ has a positive root-$r$ witness path by
\cref{lem:witness-paths}, so both endpoints of $P$ lie in the connected
support $D_r$.  Buying the five edges of that support therefore connects them.
The total cost is at most $5q/2=c(x)$.  For odd $q$, omitting a maximum independent set of
roots buys $(q+1)/2$ supports, of cost at most
$(1+1/q)c(x)$.  Thus the density-tight family is inexpensive for this
family-specific root-aware choice.  We do not claim that the displayed LP
point is optimal, and hence make no integrality-gap claim for the family.
\end{remark}

\section{Conclusion}

The proof retains the orientation and root label of each projected half-unit.
At every globally degree-two support vertex, half-integrality and the cut
constraints leave no freedom: the vertex is either unit-assigned and creates
excess outdegree inside one root support, or split-assigned and creates overlap
between two supports.  The resulting inequality
$\beta_C\ge L_C/2$ turns this local information into circuit rank.  Balancing
that component certificate with the global degree sum gives projected density
$5/8$ and, through the inherited recursive contraction lemma, an $8/5$
rounding.

The barrier family makes the universal lower bound on ordinary projected
density asymptotically tight and produces an $8/5-o(1)$ legal run under
adversarial maximum-density tie-breaking.  Its scope is deliberately narrow:
it is neither an $8/5$ half-integral integrality-gap example nor a proof that
every deterministic tie rule fails.  Possible routes beyond $8/5$ include a
better choice among density maximizers, a second efficiently computable
root-aware certificate, or a rounding step that buys an object other than one
MST of one projected set.

\section*{Acknowledgments and disclosure of AI use}
Generative AI tools, including ChatGPT and Claude, were used extensively throughout this work.  They
contributed to the formulation of theorems, algorithms, proofs, and other
claims.  They were also used for literature search and citation preparation,
figure production, verification scripts, and the drafting and iterative
revision of the manuscript.  The author guided and reviewed this process.
The scripts served only as finite sanity checks; no theorem relies on computational verification.
The author assumes responsibility for all content.

\begingroup
\setstretch{1.0}
\makeatletter\renewcommand\@openbib@code{\setlength{\itemsep}{2.5pt}\setlength{\parsep}{0pt}}\makeatother
\bibliographystyle{plain}
\bibliography{references}

@article{AKR95,
  author  = {Agrawal, Ajit and Klein, Philip N. and Ravi, R.},
  title   = {When Trees Collide: An Approximation Algorithm for the Generalized {Steiner} Problem on Networks},
  journal = {SIAM Journal on Computing},
  volume  = {24},
  number  = {3},
  pages   = {440--456},
  year    = {1995},
  doi     = {10.1137/S0097539792236237}
}

@inproceedings{AGHJM25,
  author    = {Ahmadi, Ali and Gholami, Iman and Hajiaghayi, MohammadTaghi and Jabbarzade, Peyman and Mahdavi, Mohammad},
  title     = {Breaking a Long-Standing Barrier: {$2-\varepsilon$} Approximation for {Steiner Forest}},
  booktitle = {66th IEEE Symposium on Foundations of Computer Science (FOCS)},
  pages     = {373--444},
  year      = {2025},
  doi       = {10.1109/FOCS63196.2025.00023}
}

@inproceedings{BGT24,
  author    = {Byrka, Jaros{\l}aw and Grandoni, Fabrizio and Traub, Vera},
  title     = {The Bidirected Cut Relaxation for {Steiner Tree} Has Integrality Gap Smaller Than 2},
  booktitle = {65th IEEE Symposium on Foundations of Computer Science (FOCS)},
  pages     = {730--753},
  year      = {2024},
  doi       = {10.1109/FOCS61266.2024.00052},
  note      = {arXiv:2407.19905}
}

@inproceedings{BGT25,
  author    = {Byrka, Jaros{\l}aw and Grandoni, Fabrizio and Traub, Vera},
  title     = {On the Bidirected Cut Relaxation for {Steiner Forest}},
  booktitle = {Integer Programming and Combinatorial Optimization (IPCO)},
  series    = {Lecture Notes in Computer Science},
  pages     = {114--127},
  year      = {2025},
  doi       = {10.1007/978-3-031-93112-3_9},
  note      = {Extended abstract}
}

@article{BGT26,
  author  = {Byrka, Jaros{\l}aw and Grandoni, Fabrizio and Traub, Vera},
  title   = {On the Bidirected Cut Relaxation for {Steiner Forest}},
  journal = {Mathematical Programming},
  year    = {2026},
  doi     = {10.1007/s10107-026-02407-4},
  note    = {Series B; doi:10.1007/s10107-026-02407-4; preprint: arXiv:2412.06518 (2024)}
}

@book{Frank11,
  author    = {Frank, Andr{\'a}s},
  title     = {Connections in Combinatorial Optimization},
  publisher = {Oxford University Press},
  year      = {2011}
}

@article{GM93,
  author  = {Goemans, Michel X. and Myung, Young-Soo},
  title   = {A Catalog of {Steiner Tree} Formulations},
  journal = {Networks},
  volume  = {23},
  number  = {1},
  pages   = {19--28},
  year    = {1993},
  doi     = {10.1002/net.3230230104}
}

@inproceedings{GT26,
  author    = {Gupta, Anupam and Traub, Vera},
  title     = {{Steiner Forest}: A Simplified Better-Than-2 Approximation},
  booktitle = {Proceedings of the 58th Annual ACM Symposium on Theory of Computing (STOC)},
  pages     = {943--954},
  year      = {2026},
  doi       = {10.1145/3798129.3800808},
  note      = {arXiv:2511.18460}
}

@article{GW95,
  author  = {Goemans, Michel X. and Williamson, David P.},
  title   = {A General Approximation Technique for Constrained Forest Problems},
  journal = {SIAM Journal on Computing},
  volume  = {24},
  number  = {2},
  pages   = {296--317},
  year    = {1995},
  doi     = {10.1137/S0097539793242618}
}

@article{Jain01,
  author  = {Jain, Kamal},
  title   = {A Factor 2 Approximation Algorithm for the Generalized {Steiner} Network Problem},
  journal = {Combinatorica},
  volume  = {21},
  number  = {1},
  pages   = {39--60},
  year    = {2001},
  doi     = {10.1007/s004930170004}
}

@article{NashWilliams64,
  author  = {Nash-Williams, C. St.~J.~A.},
  title   = {Decomposition of Finite Graphs into Forests},
  journal = {Journal of the London Mathematical Society},
  volume  = {s1-39},
  number  = {1},
  pages   = {12},
  year    = {1964},
  doi     = {10.1112/jlms/s1-39.1.12}
}

@misc{PT26,
  author = {Paschmanns, Paul and Traub, Vera},
  title  = {The Bidirected Cut Relaxation for {Steiner Tree}: Better Integrality Gap Bounds and the Limits of Moat Growing},
  year   = {2026},
  doi    = {10.48550/arXiv.2602.19879},
  note   = {arXiv:2602.19879}
}

@inproceedings{RV99,
  author    = {Rajagopalan, Sridhar and Vazirani, Vijay V.},
  title     = {On the Bidirected Cut Relaxation for the Metric {Steiner Tree} Problem},
  booktitle = {Proceedings of the Tenth Annual ACM-SIAM Symposium on Discrete Algorithms (SODA)},
  pages     = {742--751},
  year      = {1999}
}

@article{Schrijver00,
  author  = {Schrijver, Alexander},
  title   = {A Combinatorial Algorithm Minimizing Submodular Functions in Strongly Polynomial Time},
  journal = {Journal of Combinatorial Theory, Series B},
  volume  = {80},
  number  = {2},
  pages   = {346--355},
  year    = {2000},
  doi     = {10.1006/jctb.2000.1989}
}

@book{Schrijver03,
  author    = {Schrijver, Alexander},
  title     = {Combinatorial Optimization: Polyhedra and Efficiency},
  publisher = {Springer},
  year      = {2003}
}

@article{SZM21,
  author  = {Schmidt, Daniel and Zey, Bernd and Margot, Fran{\c c}ois},
  title   = {Stronger {MIP} Formulations for the {Steiner} Forest Problem},
  journal = {Mathematical Programming},
  volume  = {186},
  pages   = {373--407},
  year    = {2021},
  doi     = {10.1007/s10107-019-01460-6}
}

@article{Wong84,
  author  = {Wong, Richard T.},
  title   = {A Dual Ascent Approach for {Steiner Tree} Problems on a Directed Graph},
  journal = {Mathematical Programming},
  volume  = {28},
  pages   = {271--287},
  year    = {1984},
  doi     = {10.1007/BF02612335}
}
\endgroup

\appendix

\section{A self-contained polynomial-time normalization}\label{app:preprocessing}

We prove \cref{lem:normalization} and spell out rerouting, full reduction,
the split-feasibility calculation, and an executable scan schedule.  Write
$n=|V|$.

\begin{lemma}[Root rerouting]\label{lem:root-rerouting}
From any half-integral feasible point one can eliminate all assignments to
nonterminal roots, preserving feasibility and half-integrality and without
increasing cost.
\end{lemma}

This is~\cite[Lem.~2]{BGT26}; the proof is included for completeness.

\begin{proof}
Let $r$ be a nonterminal root with $z_P^r>0$ for some $P$.  Choose
$P^*\in\arg\max_{P\in\calP}z_P^r$, put $\lambda:=z_{P^*}^r$, and choose a
terminal $v\in P^*$.  The cut constraints give a $v$--$r$ flow of value at least
$\lambda$ within capacities $x^r$.  Compute an integral maximum flow with
capacities $2x^r$, decompose it into paths and circulations, retain exactly
$2\lambda$ units of path flow, and halve.  The resulting half-integral flow
$f$ has value exactly $\lambda$.  Reverse this flow in the root-$r$
layer:
\[
  \bar x^r_{ab}:=x^r_{ab}-f_{ab}+f_{ba}.
\]
Note that for every demand $P$ we have $z_P^v+z_P^r\le\sum_{s\in V}z_P^s=1$, so the
assignment below stays within $[0,1]$.  Move the reoriented layer and all
root-$r$ assignments to root $v$:
\[
  x'^v=x^v+\bar x^r,
  \quad x'^r=0,
  \qquad
  z'^v_P=z^v_P+z^r_P,
  \quad z'^r_P=0.
\]
The cost is unchanged.

Let $S\subseteq V\setminus\{v\}$ be valid for $(v,P)$; in particular the
source $v$ lies outside $S$.  If $r\notin S$, both the source $v$ and the sink $r$ lie outside $S$, so
flow conservation gives
$\bar x^r(\delta^+(S))=x^r(\delta^+(S))\ge z_P^r$.  If $r\in S$, the source
$v$ is outside and the sink $r$ inside, so
$f(\delta^+(S))-f(\delta^-(S))=-\lambda$ and
$\bar x^r(\delta^+(S))\ge\lambda\ge z_P^r$.  Adding the old root-$v$
constraint proves feasibility.  Each application eliminates one
nonterminal participating root and creates none, so at most $|V|$ applications are
needed.
\end{proof}

\begin{proof}[Proof of \cref{lem:normalization}]
Apply \cref{lem:root-rerouting} first.  If $\calP=\varnothing$,
coordinatewise reduction deletes all arc mass and we are done.  Below, as
usual, a minimum over an empty family is interpreted as $+\infty$.

\paragraph{Coordinatewise reduction.}
For a positive coordinate $x^r_{uv}$, the maximum feasible decrease is
\begin{equation}
\min\left\{x^r_{uv},\
\min_{P\in\calP}\
\min_{\substack{S:\ u\in S,\ v\notin S,\ r\notin S\\P\cap S\ne\varnothing}}
\bigl(x^r(\delta^+(S))-z_P^r\bigr)\right\}.
\label{eq:coordinate-reduction}
\end{equation}
For fixed $P$, force in turn either endpoint of $P$ into $S$; the inner
minimum is then a directed minimum cut with the vertices $u$, $v$, $r$ and the
chosen endpoint forced to prescribed sides.  We use throughout the convention
that a forced-side pattern assigning one vertex to both sides is infeasible
and contributes the value $+\infty$, while duplicate compatible requirements
are merged; if both endpoint choices are infeasible, the family contains no
valid cut and the inner minimum is $+\infty$.  All values remain
half-integral.

We use the following monotonicity.  If decreasing coordinate $a$ by
$\varepsilon$ is infeasible at $x$, then it remains infeasible at every
$x'\le x$ with $x'_a=x_a$, because every relevant cut slack can only
shrink.  Consequently, scanning the currently positive coordinates once and
reducing each by its current maximum feasible amount produces a fully reduced
point.  Every remaining positive coordinate is at most one: it belongs to a
tight cut whose right-hand side is an assignment coordinate at most one.  The
potential
\[
  \Phi:=\sum_{r,a}x_a^r
\]
is therefore at most $|V||\vec E|=O(n^3)$ after this first reduction.

\paragraph{Directed splitting off.}
Consider a split at pairwise distinct $u,v,w$ in root $r$.  For a set
$S\subseteq V\setminus\{r\}$, write
$\chi_u=\mathbf1[u\in S]$, $\chi_v=\mathbf1[v\in S]$, and
$\chi_w=\mathbf1[w\in S]$.  The change in
$x^r(\delta^+(S))$, divided by the split amount, is
\begin{equation}
  -\chi_u(1-\chi_v)-\chi_v(1-\chi_w)+\chi_u(1-\chi_w).
\label{eq:split-change}
\end{equation}
It equals $-1$ exactly in the two cases
\begin{equation}\label{eq:two-negative-families}
  v\in S,\ u,w\notin S,
  \qquad\text{or}\qquad
  u,w\in S,\ v\notin S,
\end{equation}
and is zero in the other six membership patterns.

Thus the maximum feasible split is the minimum of the two arc values and all
valid-cut slacks in these two families.  For a fixed demand and family, force
each possible demand endpoint into $S$ and impose on $u,v,w$ and $r$ the
memberships prescribed by that family; each resulting minimum is one directed
minimum cut, and the convention above applies verbatim: a pattern assigning a
vertex to both sides is infeasible and contributes $+\infty$.

Both families are needed in the exact feasibility test: by
\eqref{eq:split-change}, they are precisely the two membership patterns in
which the cut capacity drops by the split amount.  Taking the minimum over
both families and the two split-arc values therefore computes the maximum
feasible split.

The executable schedule is: reroute nonterminal roots; fully reduce all
current coordinates; repeatedly scan all roots and all ordered pairwise
distinct triples, perform a maximum positive split when one is found, and
restart the scan; after one complete scan finds no split, fully reduce once
more.  Every positive split has half-integral amount at least $1/2$ and
decreases $\Phi$ by precisely its amount.  Coordinate reduction and rerouting
do not increase $\Phi$, so at most $O(n^3)$ positive splits occur.

One scan considers $O(n^4)$ root/triple choices, and for each choice the two
families, the demand occurrences, and the at most two endpoint choices require
polynomially many directed minimum-cut computations.  Since at most $O(n^3)$
splits occur, the whole normalization performs polynomially many min-cut,
max-flow, and reduction calls.  With rational arithmetic all encoding lengths
stay polynomial in the input bit length and the number of operations, so
normalization runs in ordinary polynomial time; we make no
strong-polynomiality claim and do not optimize the exponent.

When no split remains, reduce all coordinates once more.  If a split were
feasible after this coordinatewise decrease, it would also have been
feasible beforehand because every cut slack and both split-arc values were at least as large; hence the final reduction cannot re-enable a split.  Metricity gives
$c_{uw}\le c_{uv}+c_{vw}$, so no split increases cost.  The final point is
normalized and half-integral.
\end{proof}

\section{Deferred proofs for the inherited framework}\label{app:inherited}

This appendix contains the proofs of \cref{lem:witness-paths},
\cref{lem:contraction}, and \cref{lem:degree-normal-form}.  All three belong
to the framework of Byrka, Grandoni, and Traub: \cref{lem:contraction}
combines~\cite[Lem.~4]{BGT26} with the proof of~\cite[Thm.~1]{BGT26},
\cref{lem:degree-normal-form} is~\cite[Lem.~6 and Claim~1 of its
proof]{BGT26}, and \cref{lem:witness-paths} repackages the flow reading used
in the proof of~\cite[Lem.~6]{BGT26} together with full reduction.  We
reproduce the arguments so that the paper is self-contained, not because any
of them is new here.

\subsection{Path witnesses}

\begin{proof}[Proof of \cref{lem:witness-paths}]
For each triple, max-flow/min-cut gives a flow of value at least $z_P^r$.
Scale capacities by two, compute an integral maximum flow, decompose it into
unit $v$--$r$ paths and directed circulations, retain exactly $2z_P^r$ path
units, and divide by two.  The resulting circulation-free flow has value
exactly $z_P^r$ and decomposes into simple paths.

The paths chosen for all triples of root $r$ certify its cut constraints: if
$S$ is valid for $(r,P)$, it contains an endpoint
$v\in P\setminus\{r\}$, and the chosen $v$--$r$ flow crosses $S$ with value
$z_P^r$.  If a positive root-$r$ arc occurred in none of these paths, its
coordinate could be decreased while all chosen witnesses remained feasible,
contradicting full reduction.
\end{proof}

\subsection{Densest-set contraction}

\begin{proof}[Proof of \cref{lem:contraction}]
We induct on the number of current vertices.  If no demand remains, the
procedure returns the empty forest.  Otherwise the current point is
normalized, and the hypothesis gives a set of size at least two.  Hence the
chosen maximum-density set $W$ has $|W|\ge2$, and contracting it makes strict
progress.  Put $\varrho:=\dens_x(W)$.

\smallskip
\noindent\emph{Tree charge.}
For every metric edge $uv$ internal to $W$, let
\[
  y_{uv}:=\frac1\varrho\sum_r(x_{uv}^r+x_{vu}^r).
\]
Then $y(E[W])=|W|-1$.  For every $Q\subseteq W$ with $|Q|\ge2$, maximality
of $W$ gives $x(E[Q])\le\varrho(|Q|-1)$, and hence
$y(E[Q])\le|Q|-1$.  Together with $y\ge0$, these are Edmonds' spanning-tree
polytope inequalities~\cite[Ch.~50]{Schrijver03}.  Thus
\[
  \mst(W)\le c(y)
  =\frac{c(x|_{E[W]})}{\varrho}
  \le\frac{c(x|_{E[W]})}{\eta}.
\]

\smallskip
\noindent\emph{Quotient point.}
Let $\pi:V\to V'$ contract $W$ to one vertex $\widehat w$ and fix every
other vertex.  Regard the demand list as a list of indexed occurrences
$P_j=\{s_j,t_j\}$.  Delete occurrence $j$ if
$\pi(s_j)=\pi(t_j)$; otherwise keep
$P'_j=\{\pi(s_j),\pi(t_j)\}$, preserving repetitions.  For every surviving
occurrence, quotient root $q$, and quotient arc $(a,b)$ with $a\ne b$, set
\[
 z_{P'_j}^{\prime q}:=
   \sum_{r\in\pi^{-1}(q)}z_{P_j}^r,
 \qquad
 x_{ab}^{\prime q}:=
   \sum_{r\in\pi^{-1}(q)}
   \sum_{\substack{u\in\pi^{-1}(a)\\v\in\pi^{-1}(b)}}x_{uv}^r.
\]
Arcs whose endpoints lie in one fiber are discarded.

The assignment equations are preserved.  For feasibility, let
$S'\subseteq V'\setminus\{q\}$ be valid for $(q,P'_j)$ and take its full
preimage $S=\pi^{-1}(S')$.  For every $r\in\pi^{-1}(q)$, the set $S$ avoids
$r$ and contains an endpoint of $P_j$, so it is valid for $(r,P_j)$.  Summing
the corresponding old inequalities gives
$x^{\prime q}(\delta^+(S'))\ge z_{P'_j}^{\prime q}$.  Thus the quotient point
is feasible and half-integral.

Equip $V'$ with the shortest-path metric after contraction.  Every
noninternal old arc maps to a quotient edge of no larger cost, while the
internal arcs disappear.  Consequently
\[
  c(x')+c(x|_{E[W]})\le c(x).
\]
Normalizing $(x',z')$ can only decrease its cost.  The inductive solution in
the quotient therefore costs at most $c(x')/\eta$, and together with the tree
on $W$ the total is at most $c(x)/\eta$.

\smallskip
\noindent\emph{Expansion and feasibility.}
The predecessor representation described above expands every quotient edge
level by level.  Every deleted demand occurrence became internal to a fiber
and is connected by the tree bought when that fiber was contracted; every
surviving occurrence is connected by the inductive solution.  Taking the
union therefore connects all original demands.  Finally delete cycle edges
from each connected component until a forest remains; this preserves all
pairwise connections and cannot increase cost.  The recursion depth is at
most the initial number of vertices minus one.
\end{proof}

\subsection{Degree normal form}

\begin{proof}[Proof of \cref{lem:degree-normal-form}]
Let $v$ be a Steiner support vertex.  It is not a demand endpoint, and it
cannot be a participating root because normalized participating roots are
terminals.  In particular, no positive incident arc is labeled by root $v$.
By \cref{lem:witness-paths}, every positive arc incident with $v$ lies on a
simple endpoint-to-root path, which must enter and leave $v$.  Hence
$d_{\wideG}(v)\ge2$.

Suppose $d_{\wideG}(v)=2$.  By \eqref{eq:projection-degree} the incident LP
mass is exactly one.  If it is
carried by one arc of value one, any selected path using that arc requires a
second incident arc at $v$, a contradiction.  Hence it is carried by two
half-arcs.  They cannot both enter or both leave $v$, and a path passing
through $v$ forces them to have the same root label.  Write them as
$u\to v$ and $v\to w$ in root $r$.  We have $u\ne w$, for otherwise a simple
path would return immediately to its predecessor.

Split the two half-arcs into $u\to w$.  After the split, no root-$r$ mass is
incident with $v$.  Consider a valid cut $S$ whose capacity the split decreases.
By \eqref{eq:two-negative-families}, either $v\in S$ and $u,w\notin S$, or
$u,w\in S$ and $v\notin S$.  Toggle the membership of $v$.  The toggled set
is still valid because $v$ is neither a root nor a demand endpoint, and since
the only root-$r$ arcs at $v$ are $u\to v$ and $v\to w$, the old capacity of
the toggled set equals the new capacity of $S$.  Every other cut is unchanged.
Hence the split is feasible, contradicting split-freeness.  Thus every Steiner
support vertex has degree at least three.

Now let $v$ be a terminal and choose a demand $P$ incident with it.  The
singleton constraints for roots $r\ne v$ and the complementary-singleton
constraint for root $v$ give
\[
  \sum_{r\ne v}x^r(\delta^+(\{v\}))\ge1-z_P^v,
  \qquad
  x^v(\delta^-(\{v\}))\ge z_P^v.
\]
These are disjoint parts of the mass incident with $v$, so that mass is at
least one.  Hence $d_{\wideG}(v)\ge2$.
\end{proof}

\section{Details of the exact density barrier}\label{app:barrier-details}

We verify normalization of the construction and give the equal-marginal
spanning-tree distribution used in \cref{thm:barrier}.

\subsection{Tight cuts and split blockers}

The following cuts show that every half-arc in \eqref{eq:barrier-arcs} is
indispensable in the root-$r_i$ layer.

\begin{table}[ht]
\centering
\small
\renewcommand{\arraystretch}{1.12}
\begin{tabular}{@{}lll@{}}
\toprule
Arc & Tight valid set & Unique outgoing $r_i$-arc \\
\midrule
$t_i\to a_i$ & $\{t_i\}$ & $t_i\to a_i$ \\
$r_{i-1}\to a_i$ & $\{r_{i-1}\}$ & $r_{i-1}\to a_i$ \\
$a_i\to b_i$ & $\{t_i,r_{i-1},a_i\}$ & $a_i\to b_i$ \\
$t_{i-1}\to b_i$ & $\{t_{i-1}\}$ & $t_{i-1}\to b_i$ \\
$b_i\to r_i$ & $\{t_i,r_{i-1},a_i,t_{i-1},b_i\}$ & $b_i\to r_i$ \\
\bottomrule
\end{tabular}
\caption{Full-reduction certificates.  Every displayed cut contains an
endpoint of $P_i$ or $P_{i-1}$ assigned with value one half to $r_i$.}
\label{tab:tight-arcs}
\end{table}

The only same-root incoming/outgoing pairs occur at $a_i$ and $b_i$.  Each
possible split is blocked by a tight cut from the first family of
\eqref{eq:two-negative-families}.

\begin{table}[ht]
\centering
\small
\renewcommand{\arraystretch}{1.12}
\begin{tabularx}{\textwidth}{@{}l>{\raggedright\arraybackslash}p{.30\textwidth}
                                  >{\raggedright\arraybackslash}X@{}}
\toprule
Split triple $(u,v,w)$ & Blocking cut $S$ & Outgoing arc and demand endpoint \\
\midrule
$(t_i,a_i,b_i)$ & $\{a_i,r_{i-1}\}$ & $a_i\to b_i$; endpoint $r_{i-1}$ of $P_{i-1}$ \\
$(r_{i-1},a_i,b_i)$ & $\{a_i,t_i\}$ & $a_i\to b_i$; endpoint $t_i$ of $P_i$ \\
$(a_i,b_i,r_i)$ & $\{b_i,t_{i-1}\}$ & $b_i\to r_i$; endpoint $t_{i-1}$ of $P_{i-1}$ \\
$(t_{i-1},b_i,r_i)$ & $\{b_i,a_i,t_i,r_{i-1}\}$ & $b_i\to r_i$; endpoints of $P_i$ or $P_{i-1}$ \\
\bottomrule
\end{tabularx}
\caption{Split-blocking cuts.  Each has outgoing capacity exactly one half
in the root-$r_i$ layer.}
\label{tab:blocking-cuts}
\end{table}

\subsection{An equal-marginal spanning-tree distribution}

Write the five projected edge types as
\begin{align*}
 \alpha_i&=t_i a_i,&
 \beta_i&=r_{i-1}a_i,&
 \gamma_i&=a_i b_i,\\
 \delta_i&=t_{i-1}b_i,&
 \epsilon_i&=b_i r_i.
\end{align*}
Each vertex has degree three or two: $a_i$ lies on $\alpha_i,\beta_i,\gamma_i$
and $b_i$ on $\gamma_i,\delta_i,\epsilon_i$, while $t_i$ lies on
$\alpha_i,\delta_{i+1}$ and $r_i$ on $\beta_{i+1},\epsilon_i$.

\begin{lemma}[Topology of the deleted graphs]\label{lem:barrier-topology}
Let $q\ge3$ and delete from the projection all $q$ edges of one class.
\begin{enumerate}[label=(\roman*),leftmargin=1.8em]
\item If the deleted class is $\alpha$ or $\delta$, the remaining graph is
connected and unicyclic, and its unique cycle is
$r_{i-1}\,\beta_i\,a_i\,\gamma_i\,b_i\,\epsilon_i\,r_i$ closed over
$i\in\mathbb Z_q$, of length $3q$.
\item If the deleted class is $\beta$ or $\epsilon$, the remaining graph is
connected and unicyclic, and its unique cycle is
$t_i\,\alpha_i\,a_i\,\gamma_i\,b_i\,\delta_i\,t_{i-1}$ closed over
$i\in\mathbb Z_q$, of length $3q$.
\item If the deleted class is $\gamma$, the remaining graph is a disjoint
union of $\gcd(q,2)$ cycles covering all $4q$ vertices: one cycle of length
$4q$ if $q$ is odd, and two cycles of length $2q$ if $q$ is even.  In the even
case $a_i$ and $b_i$ lie in different cycles, so every $\gamma_i$ joins them.
\end{enumerate}
\end{lemma}

\begin{proof}
(i) Delete all $\alpha$.  Then $t_i$ has degree one, so all $q$ vertices $t_i$
are leaves; deleting them leaves the $3q$ vertices $r_i,a_i,b_i$, each of
degree exactly two, with the edges $\beta,\gamma,\epsilon$.  A graph in which
every degree is two is a disjoint union of cycles, and following
$r_{i-1}\to a_i\to b_i\to r_i$ advances the index by one after three edges,
so there is exactly one cycle, of length $3q$.  Reattaching the $q$ leaves
gives a connected graph with $4q$ vertices and $4q$ edges whose only cycle is
the displayed one.  Deleting all $\delta$ instead again makes every $t_i$ a
leaf, this time through $\alpha_i$, and leaves the same three classes on
$r_i,a_i,b_i$.

(ii) Delete all $\beta$.  Then $r_i$ has degree one and is a leaf; deleting
the $q$ leaves leaves $t_i,a_i,b_i$, each of degree two, with the edges
$\alpha,\gamma,\delta$.  Following $t_i\to a_i\to b_i\to t_{i-1}$ decreases
the index by one after three edges, so there is exactly one cycle, of length
$3q$.  Deleting all $\epsilon$ instead makes each $r_i$ a leaf through
$\beta_{i+1}$ and leaves the same three classes.

(iii) Delete all $\gamma$.  Every vertex now has degree exactly two, so the
graph is a disjoint union of cycles covering all $4q$ vertices and using all
$4q$ remaining edges.  Starting at $r_i$ and following the unique continuation
gives
\[
  r_i\ \xrightarrow{\ \beta_{i+1}\ } a_{i+1}
      \ \xrightarrow{\ \alpha_{i+1}\ } t_{i+1}
      \ \xrightarrow{\ \delta_{i+2}\ } b_{i+2}
      \ \xrightarrow{\ \epsilon_{i+2}\ } r_{i+2},
\]
so four edges advance the index by two.  The orbits of $i\mapsto i+2$ on
$\mathbb Z_q$ therefore correspond to the cycles: there are $\gcd(q,2)$ of
them, each of size $q/\gcd(q,2)$, and each contributes four edges per index
step, giving cycle length $4q/\gcd(q,2)$.  For odd $q$ this is one cycle of
length $4q$; for even $q$ it is two cycles of length $2q$.  In the even case
the displayed walk shows that the cycle through $r_i$ contains $a_{i+1}$ and
$b_{i+2}$, so $a_j$ lies in the cycle of the residue $j-1$ and $b_j$ in the
cycle of the residue $j$ modulo two.  These residues differ, so $a_j$ and
$b_j$ lie in different cycles and $\gamma_j={a_j b_j}$ joins them.
\end{proof}

For $X\in\{\alpha,\beta,\delta,\epsilon\}$, delete all $q$ edges of class
$X$ and then one uniformly random edge of the unique cycle supplied by
\cref{lem:barrier-topology}(i)--(ii).  The result has $4q-1$ edges, is
connected, and is acyclic, hence a spanning tree; call it an $X$-tree.  By
\cref{lem:barrier-topology}, that cycle consists of the classes
$\beta,\gamma,\epsilon$ when $X\in\{\alpha,\delta\}$ and of
$\alpha,\gamma,\delta$ when $X\in\{\beta,\epsilon\}$.

Suppose first that $q$ is odd.  By \cref{lem:barrier-topology}(iii), deleting
all $\gamma$-edges leaves one $4q$-cycle; delete one uniformly random edge to
obtain a $\gamma$-tree.
Give each of the four side-tree distributions probability
\[
  x_o:=\frac{12q-3}{60q-20}
\]
and the $\gamma$-tree distribution probability
\[
  y_o:=\frac{3q-2}{15q-5}.
\]
Both quantities are positive for every odd $q\ge3$ and satisfy
$4x_o+y_o=1$, so they define a probability distribution over the five
families.  A fixed non-$\gamma$ edge is deleted with probability
one in exactly one side family, lies on the unique cycle in two side
families, and lies on the $\gamma$-cycle.  Its deletion probability is
\[
  x_o+\frac{2x_o}{3q}+\frac{y_o}{4q}=\frac{q+1}{5q}.
\]
A fixed $\gamma$-edge is deleted with probability
\[
  y_o+\frac{4x_o}{3q}=\frac{q+1}{5q}.
\]

Now suppose $q$ is even, so $q\ge4$.  By \cref{lem:barrier-topology}(iii),
deleting all $\gamma$-edges leaves two disjoint cycles of length $2q$, and
every $\gamma_i$ joins them.  Retain one uniformly random $\gamma$-edge as a
bridge and delete one uniformly random edge from each cycle; the result has
$4q-2+1=4q-1$ edges, is connected, and is acyclic, hence a spanning tree.  A fixed $\gamma$-edge is then deleted with probability
$(q-1)/q$, while a fixed non-$\gamma$ edge is deleted with probability
$1/(2q)$.  Give each side-tree distribution probability
\[
  x_e:=\frac{6q-9}{30q-40}
\]
and this $\gamma$-tree distribution probability
\[
  y_e:=\frac{3q-2}{15q-20}.
\]
Both quantities are positive for every even $q\ge4$ and satisfy
$4x_e+y_e=1$.  Every non-$\gamma$ edge has deletion probability
\[
  x_e+\frac{2x_e}{3q}+\frac{y_e}{2q}=\frac{q+1}{5q},
\]
and every $\gamma$-edge has deletion probability
\[
  \frac{4x_e}{3q}+\frac{q-1}{q}y_e=\frac{q+1}{5q}.
\]
Thus every edge is present with probability
\[
  1-\frac{q+1}{5q}=\frac{4q-1}{5q}.
\]
As a sanity check, $5q\cdot(4q-1)/(5q)=4q-1$, the number of edges in every
spanning tree.

\end{document}